\documentclass[letterpaper, 10 pt, conference]{ieeeconf} 
\usepackage{cite}
\usepackage[cmex10]{amsmath}
\usepackage{algorithm}
\usepackage{algorithmic}
\usepackage[inline]{trackchanges}

\usepackage[keeplastbox]{flushend}

\addeditor{swanand}
\addeditor{alex}

\usepackage{amsfonts, amssymb, amscd, xspace}

\usepackage{color}
\usepackage{graphicx}
\usepackage{epstopdf}

\newtheorem{theorem}{Theorem}
\newtheorem{lemma}{Lemma}

\newtheorem{remark}{Remark}
\newtheorem{construction}{Construction}

\newtheorem{definition}{Definition}

\newcommand{\etal}{{\it et al.}}
\newcommand{\ie}{{\it i.e.}}
\newcommand{\eg}{{\it e.g.}}

\newcommand{\rank}{\textrm{rank}}
\newcommand{\twomatrix}[2]{\begin{bmatrix} #1 \\ #2\end{bmatrix}}

\newcommand{\Sf}{S} 
\newcommand{\X}{X} 
\newcommand{\al}{\alpha} 
\newcommand{\be}{\beta} 
\newcommand{\GF}[1]{\mathbb{F}_{#1}} 
\newcommand{\E}{E} 
\newcommand{\G}{\mathcal{G}} 
\newcommand{\SG}{S_{\mathcal{G}}} 
\newcommand{\Bs}{B_s} 

\newcommand{\ELL}{\mathcal{L}} 
\newcommand{\CL}{C_{\mathcal{L}}} 

\newcommand{\I}[2]{I\left(#1;#2\right)} 
\newcommand{\code}{\mathcal{C}} 

\newcommand{\qr}{q_r}
\newcommand{\Pij}[2]{P_{#1,#2}} 
\newcommand{\pij}[2]{p_{#1,#2}} 
\newcommand{\Gij}[2]{G_{#1,#2}} 

\newcommand\blfootnote[1]{%
  \begingroup
  \renewcommand\thefootnote{}\footnote{#1}%
  \addtocounter{footnote}{-1}%
  \endgroup
}

\title{\LARGE \bf
Universally Weakly Secure Coset Coding Schemes for Minimum Storage Regenerating (MSR) Codes}

\author{{Swanand Kadhe and Alex Sprintson}}

\begin{document}

\maketitle

\begin{abstract}
We consider the problem of designing codes for distributed storage that protect user data against eavesdroppers that can gain access to network links as well as individual nodes. Our goal is to achieve weak security (also known as block security) that requires that the eavesdroppers would not be able to decode individual files or combinations of a small number of files. The standard approach for achieving block security is to use a joint design scheme that consists of (inner) storage code and the (outer) coset code. However, jointly designing the codes requires that the user, who pre-processes and stores the files, should know the underlying storage code in order to design the (outer) linear transformation for achieving weak security. In many practical scenarios, such as storing the files on the third party cloud storage system, it may not be possible for the user to know the underlying storage code.
  
In this work, we present universal schemes that separate the outer code design from the storage code design for minimum storage regenerating codes (MSR). Our schemes allow the independent design of the storage code and the outer code. Our schemes use small field size and can be used in a broad range of practical settings.
\end{abstract}


\blfootnote{Swanand Kadhe and Alex Sprintson are with the Department of Electrical and Computer Engineering at Texas A\&M University, USA; Emails:\{swanand.kadhe, spalex\}@tamu.edu.\\
}

%

\section{Introduction}
Coding for distributed storage systems (DSS) has recently received significant attention from the research community. The main focus has been on designing and analyzing novel erasure codes that efficiently handle node failures in distributed storage systems, see, \eg,~\cite{Dimakis:10,Dimakis-Survey:11,Gopalan:11,Papailiopoulos:14}.

An important challenge for a DSS is providing secrecy against eavesdropping. The problem with using conventional secret key-based encryption techniques is that they require secret key management mechanisms, which incur significant computational and communication overheads in distributed settings.  
Following the work of~\cite{Pawar:10, Pawar:11}, a number of investigations have been carried out on information-theoretically securing the regenerating codes, see, \eg,~\cite{Shah:11,Zhu:13,Rawat:14,Goparaju:15,Mazumdar:16,Rawat:16,KadheS:17}.

Most of the security results in distributed storage community are focused on the paradigm of information-theoretic {\it perfect secrecy}. Intuitively, perfect secrecy requires that the eavesdropper gains absolutely {\it no information} about the stored data from its observations. 
To be precise, suppose that a DSS is storing $\Bs$ data files $\Sf = \{S_1,\ldots,S_{\Bs}\}$, where each file can be considered as a symbol in a finite field $\GF{q}$. Let $\E$ denote the set of (encoded) files that an eavesdropper {\it Eve} can observe. A DSS is said to be \emph{perfectly secure} if the mutual information between the message symbols $\Sf$ and the eavesdropped symbols $\E$ is zero, \ie, $\I{\Sf}{\E} = 0$. 

For many practical storage systems, perfect secrecy condition might be too strong. Moreover, coding schemes that provide perfect secrecy involve mixing data symbols with random keys to confuse the eavesdropper, which incurs loss in the storage capacity. Considering these drawbacks of the perfect secrecy notion, we focus on the notion of {\it weak security} proposed by Bhattad and Narayanan~\cite{Bhattad:05}.

The weak security condition requires that Eve cannot gain any information about {\it any} group of files of size $g$, where $g$ is some positive integer. Based on the premise that individual files carry meaningful information, the motivation behind weak security is that, even if Eve obtains any $g-1$ files as a side information, she cannot decode for any new file. 
For example, let the number of files be \mbox{$\Bs = 4$,} and suppose the files are chosen independently and uniformly at random over $\GF{5}$. 
Suppose that Eve observes the following two encoded symbols \mbox{$\E  = \{S_1 + S_2 + S_3 + S_4,\: S_1 + 2S_2 + 3S_3 + 4S_4\}$.} 
The scheme protects any group of $g = 2$ files. This ensures that, even if Eve has a side-information of any one file, she cannot decode for any other file by observing $E$. Note that, when $g = 1$, weak security requires that Eve gains no information about any individual file, \ie, $\I{S_i}{E} = 0$ $\forall i$.

It was observed in~\cite{Bhattad:05} that weakly secure coding schemes do not incur loss in the capacity, as it is not required to mix any {\it private randomness}. Essentially, a weakly secure scheme protects a group of files by using the other files as random keys.

Note that the notion of weak security that is introduced in~\cite{Bhattad:05} and considered throughout this paper, is different from the conventional notion of information-theoretic {\it weak secrecy}, which is defined for asymptotically large block-lengths. The weak security notion considered in this paper is applicable to finite block-lengths as well. The notion of weak security has also been referred to as {\it block security}, as it requires protecting blocks of information of different sizes (see, \eg~\cite{Dau:13}). 

Despite of its practical benefits, there have been relatively very few attempts on employing weak security for DSS. In~\cite{Oliveira:12}, Oliveira \etal~have presented a construction of weakly secure erasure codes for DSS without considering the regeneration aspects. Dau \etal~\cite{Dau:13} analyzed the weak security properties of two families of regenerating codes: regular-graph codes~\cite{Rashmi:09} and product-matrix codes~\cite{Rashmi:11}. Going a step ahead, in~\cite{Kadhe:14,Kadhe:14Allerton}, we presented outer code constructions that weakly secure product-matrix (PM) codes~\cite{Rashmi:11}. 



In all these solutions, the standard approach for achieving weak security is to jointly design an outer coset code and an inner storage code. The main limitation of jointly designing the codes is that it requires the user, who  designs the (outer) linear transformation for achieving weak security, to know the underlying storage code. In many practical scenarios, such as storing the files on the third party cloud storage system, it may not be possible for the user to know the underlying storage code.

In this paper, we focus on {\it universal} schemes that separate the outer code design from the storage code design. Such a universal approach was first proposed in~\cite{Silva:08} (see also~\cite{Silva:11}) to achieve perfect security in network coding. This approach was extended for weakly securing network codes in~\cite{Silva:09,Kurihara:12}, and was adapted for weakly securing distributed storage codes, namely regenerating codes, in~\cite{Kurihara:13}. The idea in these works is to design an outer code based on rank-metric codes~\cite{Gabidulin:85}. However, the main drawback of using a rank-metric code is that the required field size is significantly large. In particular, the universal outer code of~\cite{Kurihara:13} requires the field size of $q^{2B}$, where $q$ is the field size of the underlying storage code and $B$ is the total number of information symbols stored.

The question we ask in this paper is that if, instead of designing a universal outer code for weakly securing {\it any} storage code, if we restrict to securing a particular class of storage codes, can we design outer codes over small field size. We answer this question affirmatively by considering an important class of storage codes called minimum storage regenerating (MSR) codes~\cite{Dimakis:10}. Essentially, an MSR code is a maximum distance separable (MDS) code that minimizes the amount of data  downloaded while repairing a failed node. 

{\it Our Contributions:} First, we present a construction of universal outer code that can achieve weak security of individual symbols, \ie, $g = 1$, in any $(n,k)$ MSR code against an eavesdropper that can observe any $k-1$ storage nodes. The required field size of the scheme is $O(B^k)$, where $B$ is the number of stored information symbols. Next, we present a construction of universal outer code that can achieve weak security with maximum possible $g$ in any MSR code against an eavesdropper that can observe any single storage node. The required field size of the scheme is $O(B^{\alpha})$, where $B$ is the total number of stored information symbols and $\alpha$ is the number of symbols stored on each node.

\section{Preliminaries}
\label{sec:preliminaries}

\subsection{Regenerating Codes}
\label{sec:regen-codes}
Consider a DSS that stores a set of $B$ files given as $\Sf = \{S_1,\ldots,S_B\}$, where each file can be considered as a uniformly and independently drawn symbol from a finite field $\GF{q}$. The system contains $n$ storage nodes, with each node capable of storing $\alpha$ files. An $(n,k,d,\alpha,\beta)$ regenerating code encodes the $B$ files into $n\alpha$ files over $\GF{q}$ in such a way that it satisfies the following two properties:
(i) {\it reconstruction property} -- a {\it data collector} (DC) connecting to \emph{any} $k$ out of $n$ nodes can reconstruct the entire set of files; (ii) {\it regeneration property} -- when a storage node  fails, it can be \emph{regenerated} by adding a new node which downloads $\beta$ symbols each from any $d$ out of the remaining $n - 1$ nodes. The $d$ nodes participating in node repair are referred to as the \emph{helper nodes}, and the $d\beta$ number of symbols downloaded is referred to as the \emph{repair bandwidth}. 

Using the cut-set bounds, the capacity of an optimal $(n,k,d,\al,\be)$ regenerating code is bounded  as~\cite{Dimakis:10}
\begin{equation}
\label{eq:capacity}
B \leq \sum_{i=0}^{k-1}\min\{\alpha,(d-i)\beta\}.
\end{equation}
It is easy to see that there is trade-off between storage space per node $\alpha$ and repair bandwidth $d\beta$. Most of the results in the literature focus on the two extreme points of the optimal storage-repair bandwidth trade-off curve. The codes on one extreme point that minimize the repair bandwidth first and then the storage per node are referred to as {\it Minimum Bandwidth Regenerating} (MBR) codes; whereas, the codes on the other extreme point that first minimize the storage per node and then the repair bandwidth are referred to as {\it Minimum Storage Regenerating} (MSR) codes. Several explicit code constructions have been proposed for these extreme points for {\it exact} repair model, wherein the repaired node is an exact replica of the failed node (see, \eg,~\cite{Dimakis-Survey:11,Rashmi:11,YeB:17, Sasidharan:17} and references therein). In this paper, we focus on the codes at the MSR point. Note that, for MSR codes, $B = k\al$, and these codes can be considered as MDS codes with minimum repair bandwidth. 

\subsection{Eavesdropper Model}
\label{sec:Eve-model}
We assume that an eavesdropper {\it Eve} can access the data stored in any $\ell$ $(< k)$ storage nodes. 
Further, we assume that Eve is passive, has unbounded computational power, and has the knowledge of the coding scheme being used.

It is worth pointing out that, for an MSR code, the number of downloaded symbols $(d\beta)$ is strictly greater than the number of stored symbols $(\alpha)$. 
Therefore, Eve can potentially gain more information by observing the data downloaded during node repair than merely observing the data stored on the node. This motivates a generalized eavesdropper model for a DSS, called as the $(l_1,l_2)$-eavesdropper model, where, Eve can access the data stored on any $l_1$ nodes, and the data downloaded during the regeneration of any $l_2$ nodes (see~\cite{Shah:11,Rawat:14}). Our focus is on the case $l_1 = \ell$, $l_2=0$.

\subsection{Information-theoretic Secrecy}
\label{sec:secrecy}
Suppose we need to store a set $\Sf$ of $\Bs$ files {\it securely}, where $\Bs \leq B$. 
Let $\E$ denote the set of (coded) files observed by Eve. A DSS is said to be {\it perfectly secure} if $\I{\Sf}{\E} = 0$. Under this requirement, Pawar {\it et al.}~\cite{Pawar:10} characterized an upper bound on the {\it secrecy capacity} as:
\begin{equation}
\label{eq:capacity-secure}
\Bs \leq \sum_{i=\ell}^{k-1}\min\{\alpha,(d-i)\beta\}.
\end{equation}
Comparing \eqref{eq:capacity} and \eqref{eq:capacity-secure}, we can say that in a perfectly secure DSS, the $\ell$ nodes that are accessed by the eavesdropper cannot effectively contain any useful information. Consequently, the perfect secrecy requirement results in a loss of storage capacity, {\it i.e.}, $\Bs < B$. 


In this paper, we focus on a relaxed, yet practically appealing notion of {\it weak security}~\cite{Bhattad:05}. The weak security condition demands that all small groups of files of bounded size are protected from the eavesdropper. The bound on group size is given by a parameter $g$. We define the weak security in the following.

\begin{definition}
\label{def:weak-security} 
Let $S = \left[S_1 \:\: S_2\:\: \cdots \:\: S_{\Bs}\right]$ be the set of $\Bs$ files, each one chosen independently and uniformly at random over some finite field $\GF{q}$. Let \mbox{$\SG := \{S_i : i\in\G\}$} for some set $\G\subset[\Bs]$, where $[\Bs] := \{1,\ldots,\Bs\}$. Consider a DSS encoding $S$ into $n\al$ codes files as \mbox{$C = f(S)$,} where \mbox{$f:\GF{q}^{\Bs} \rightarrow \GF{q}^{n\al}$} is some (potentially stochastic) encoder. The coded files $C$ are stored across $n$ nodes as \mbox{$C = \left[C_1\:\: C_2\:\: \cdots \:\: C_n\right]$,} where $C_i$ is the set of $\al$ files stored on node $i$. Suppose Eve observes a set $\ELL\subset[n]$ of $\ell$ storage nodes and let $\CL = \cup_{i\in\ELL}C_{i}$. Then, the encoder $f$ is said to be {\it $g$-weakly secure} against an eavesdropper of strength $\ell$, if, for every $\ELL\subset[n]$ such that $|\ELL|\leq\ell$, we have, 
\begin{equation}
\label{eq:g-block-secure}
\I{\SG}{E}  =  0 \:\: \forall{\G}\subseteq[\Bs] : \left|\G\right| \leq g.
\end{equation}
\end{definition}
\vspace{2mm}

Note that when $g = 1$, we have $\I{S_i}{\CL} = 0$, $\forall i\in[\Bs]$. Assuming that individual files carry meaningful information, $1$-weak security prevents eavesdropper from gaining any meaningful information. 

It is easy to verify that condition~\eqref{eq:g-block-secure} is equivalent to the following condition~\cite[Proposition~5]{Silva:09}:
\begin{equation}
\label{eq:g-secure}
\I{S_i}{\CL | \SG}  =  0 \:\: \forall i\in[\Bs]\setminus\G,\: \forall\G\subset[\Bs] : \left|\G\right| \leq g-1.
\end{equation}
In other words, when the system is $g$-weakly secure, even if Eve obtains any $g-1$ files as a side information, she cannot decode any additional file. We note that, in~\cite{Bhattad:05}, the weak security notion was proposed using condition~\eqref{eq:g-secure}.

\begin{remark}
\label{rem:max-g}
As noted in~\cite{Silva:09}, $g$-weak security is equivalent to perfect security of a {\it message} $S' = \SG \subset \Sf$ for any $\G$ of size up to $g$ (see~\eqref{eq:g-block-secure}). In particular, if we treat $\SG$ as a message and the rest of the symbols as random keys (for any $\G$ of size up to $g$), then \eqref{eq:g-block-secure} is equivalent to the perfect secrecy of $\SG$. Therefore, it is possible to store $B$ symbols with $g$-weak security against an eavesdropper observing any $\mu$ coded symbols, only if $g\leq B - \mu$. 
\end{remark}

\section{Outer Coset Code for Weak Security}
\label{sec:securing-PM-MSR}
Our approach to achieve weak security is to use an outer code construction based on {\it coset coding}~\cite{Ozarow:84}. We briefly review the coset coding in this section. A coset code is constructed using a $(B,B-\Bs)$ linear code $\code_s$ over $\GF{q}$ with parity-check matrix $H\in\GF{q}^{\Bs\times B}$. Specifically, the set of $\Bs$ files $\Sf$ is encoded by selecting uniformly at random some $\X\in\GF{q}^B$ such that $\Sf = H\X$. In other words, the vector $\Sf$ can be considered as a syndrome specifying a coset of $\code_s$, and the codeword $\X$ is a randomly chosen element of that coset. 

Next, the codeword $X$ is encoded using a regenerating code (as an inner code) to obtain $C\in\GF{q}^{n\al}$, \ie, $C = G\X$, where $G\in\GF{q}^{n\alpha\times B}$ is a generator matrix of the regenerating code. To obtain the information symbols $\Sf$, a user needs to first decode the regenerating code to get $\X$, 
and then, decode the outer coset code to get $\Sf$. The decoding operation of a coset code consists of simply computing the syndrome $\Sf = H\X$. The node repair process is inherited from the regenerating code.

For the rest of the paper, for simplicity, we refer to $H$ as a coset code and $G$ as a storage/MSR code.

To design the matrix $H$ appropriately, we need to transform the weak security condition~\eqref{eq:g-block-secure} into a condition involving $H$. For this, we use the following result from~\cite[Lemma 6]{Silva:11}, which is a generalization of~\cite[Theorem 1]{Salim:07}.
\begin{lemma}
\label{lem:silva}
(\cite{Silva:11}) Suppose that a coset code with parity-check matrix $H\in\GF{q}^{\Bs\times B}$ is used as an outer code over a storage code with generator matrix $G$ to store the message \mbox{$\Sf = [S_1\:\cdots\: S_{\Bs}]$.} Suppose each symbol $S_i$ for $i \in [\Bs]$ is chosen independently and uniformly at random over some finite alphabet. Let $\CL = G'\X$ be the $\mu$ symbols observed by an eavesdropper, where $G'$ is a $\mu\times B$ sub-matrix of $G$. Then, for any $\G\subseteq[\Bs]$ such that $\left|\G\right| \leq B - \mu$, we have
\begin{equation}
\label{eq:subspace-condition-0}
\I{\SG}{\CL} = \rank\: H_{\G} + \rank\: G' - \rank\:\twomatrix{H_{\G}}{G'},
\end{equation}
where $H_{\G}$ is a sub-matrix of $H$ formed by choosing the rows indexed by the set $\G$.
\end{lemma}

Then, using~\eqref{eq:g-block-secure} and~\eqref{eq:subspace-condition-0}, it follows that a universal coset code $H$ ensures that
\begin{equation}
\label{eq:subspace-condition}
\rank\:\begin{bmatrix}H_{\G} \\ G' \end{bmatrix} = \rank\: {H_{\G}} + \: \rank\: G',
\end{equation}
for every $\G\subset[\Bs]$ such that $\left|\G\right| \leq g$ and for every storage code $G$. (Recall that $G'$ is a $\mu\times B$ sub-matrix of $G$ corresponding to $\mu$ eavesdropped symbols.)

Such a coset code $H$ was first constructed in~\cite{Silva:09} for $g\leq 2$ using a rank-metric code over $\GF{q^B}$ to secure any $G$ over $\GF{q}$ with application to network coding. This construction was extended for $g\leq \Bs-\mu$ in~\cite{Kurihara:12} again using rank-metric codes, requiring the field size of $q^{2B}$, where $q$ is the field size for the entries of $G$. The authors of~\cite{Kurihara:13} adapted the construction of~\cite{Kurihara:12} to weakly secure any regenerating code. 
The main drawback of such an outer code based on a rank-metric code is its high field size. Instead of securing any regenerating code, we restrict our attention to the class of MSR codes, and present universal outer code constructions over small field size in the next section.


\section{Universal Outer Codes for MSR Codes}
\label{sec:main}

In the following, we present constructions for universal outer codes to achieve $g$-weak security in any $(n,k,d,\al,\be)$-MSR code. 
In particular, we consider the following two scenarios: (i) maximum $\ell$ and minium $g$, \ie, $\ell = k - 1$ and $g = 1$, and (ii) minimum $\ell$ and maximum $g$, \ie, $\ell = 1$ and $g = B - \al$. We assume that the user only knows the code parameters $n,k,d,\al,\be$, and $q$. In addition, we assume that the encoding of the MSR code is systematic. We begin with setting up necessary notation for MSR codes.

{\it Notation for MSR Codes:} Consider an $(n,k,d,\al,\be)$-MSR code $\code$ over $\GF{q}$, storing $B = k\al$ information symbols (see Sec.~\ref{sec:regen-codes}). Let \mbox{$G  = \left[G_1^T \:\: G_2^T \:\: \cdots \:\: G_n^T\right]^T$} be an $n\al\times B$ generator matrix of $\code$, where $G_i$ is an $\al\times B$ matrix corresponding to the symbols stored on node $i$. We refer to $G_i$ as a generator matrix of node $i$. Let us denote $G_i$ as \mbox{$G_i = \left[\Gij{i}{1} \:\: \Gij{i}{2} \:\: \cdots \:\: \Gij{i}{\al}\right]$,} where $\Gij{i}{j}$ is an $\al\times\al$ matrix. 

We assume that $G$ is in systematic form, and the first $k$ nodes are systematic. In other words, we have $\Gij{i}{i} = I_{\al}$ and $\Gij{i}{j} = 0_{\al}$ for $1\leq i,j\leq k$ such that $i\neq j$, where $I_{\al}$ is an $\al\times\al$ identity matrix and $0_{\al}$ is an $\al\times\al$ zero matrix. 

For any matrix (or vector) $W$ with $B = k\al$ columns, we refer to the $\al$ columns of $W$ indexed from $(j-1)\al+1$ through $j\al$, as the $j$-th {\it thick-column} of $W$ $(1\leq j\leq k)$.

\subsection{Construction for $\ell = k-1$ and $g = 1$}

Note that $\ell = k - 1$ is the maximum possible strength that Eve can have for an $(n,k)$-MSR code, as any $k$ nodes recover the entire stored data. 
The motivation behind \mbox{$g=1$} is to protect every individual file, which usually carry meaningful information. The idea for constructing $H$ is to begin with a Vandermonde matrix over some base field and then scale some of its appropriately chosen columns by elements lying in an extension field. The details are given in the following.

\begin{construction}
\label{con:universal-g-1}
Consider the parameters of an MSR code as $n, k, d, \al, \be$ and $q$. Choose $\qr$ as the smallest power of $q$ greater than equal to $B = k\al$. For $1\leq i,j\leq B$, choose the entries $h_{i,j}$ of $H$ as follows:
\begin{equation}
\label{eq:universal-g-1}
h_{i,j}  = 
\left\{
\begin{array}{lr}
\omega^{\frac{j}{\alpha}}\beta_j^{i-1} & \textrm{if} \:\: \alpha\mid j,\\
\beta_j^{i-1} & \textrm{otherwise},
\end{array}
\right.
\end{equation}
where $\beta_1, \cdots, \beta_B$ are $B$ distinct elements of $\GF{\qr}$, and $\omega$ is a primitive element of $\GF{\qr^{k+1}}$.
\end{construction}

Next, we show that the above construction can universally achieve $1$-weak security for any MSR code for $\ell = k-1$.

\begin{theorem}
\label{thm:universal-g-1}
The outer coset code of Construction~\ref{con:universal-g-1} can be used universally with any $(n,k)$-MSR code to store \mbox{$\Bs = k\al$} symbols over $\GF{\qr^{k+1}}$ with $1$-weak security against an eavesdropper observing any $\ell = k-1$ nodes, each storing $\al$ coded symbols.
\end{theorem}
\begin{proof}
First, note that, since $H$ is a $B\times B$ Vandermonde matrix with some of its columns scaled, it is non-singular, resulting in $\Bs = B = k\al$.

Next, we prove the $1$-weak security. Let $H_i$ denote the $i$-th row of $H$. Let the set of nodes accessed by Eve be $\ELL = \{i_1,\cdots,i_{\ell}\}$ and let $\CL = \{C_i : i\in\ELL\}$. For $g = 1$, we want to prove $\I{S_i}{\CL} = 0$ for every $i\in[B]$ (see~\eqref{eq:g-block-secure}). From \eqref{eq:subspace-condition-0}, this is equivalent to showing that, for every $i\in[B]$, the following matrix is full-rank:
\begin{equation}
\label{eq:T-got-g-1} 
T = \twomatrix{H_i}{G'} = 
\begin{bmatrix}
H_i\\
G_{i_1}\\
G_{i_2} \\
\vdots \\ 
G_{i_{\ell}}
\end{bmatrix}
\end{equation}
Towards proving this, we consider the following two cases.

{\it Case 1:} $\ELL \subset[k]$. In other words, all the nodes Eve observes are systematic. 
Let $j = [k]\setminus\ELL$. Then, the $j$-th thick-column of $G'$ 
is zero. Clearly, $H_i$ cannot be in the row space of $G'$.

{\it Case 2:} $\ELL\not\subset[k]$. In this case, at least one parity node is eavesdropped. Arbitrarily choose an index $j$ such that $j\in[k]\setminus\ELL$. Note that there are at least two such systematic nodes not in $\ELL$. Due to the reconstruction property of MSR codes, $\twomatrix{G'}{G_j}$ should be invertible. As the $j$-th thick-column of $G_j$ is an identity matrix and all its other thick-columns are zero, this implies that the $(k-1)\al\times(k-1)\al$ matrix, say  $G''$, formed by all the thick-columns of $G'$ except the $j$-th one is invertible. Thus, we can perform row operations on $T$ to obtain the following matrix: 
$$T_1 = \twomatrix{H_i}{(G'')^{-1}G'}.$$
It is easy to see that, after reordering the rows and columns of $T_1$, we can get
\begin{equation}
\label{eq:T-semifinal}
T_2 = 
\begin{bmatrix}
H_{i,1} & H_{i,2} & \cdots & H_{i,k-1} & H_{i,k}\\
I_{\al} & 0_{\al} & \cdots & 0_{\al} & \Pij{1}{\al}\\
0_{\al} & I_{\al} & \cdots & 0_{\al} & \Pij{2}{\al}\\
\vdots & \vdots & \ddots & \vdots & \vdots\\
0_{\al} & 0_{\al} & \cdots & I_{\al} & \Pij{k-1}{\al}
\end{bmatrix},
\end{equation}
where \mbox{$H_{i,j} = [\beta_{(j-1)\al+1}^{i-1} \: \beta_{(j-1)\al+2}^{i-1} \: \cdots \: \beta_{j\al-1}^{i-1} \: \omega^j\beta_{j\al}^{i-1}]$,} and $\Pij{i}{j}\in\GF{\qr}^{\al\times\al}$. Now, we obtain a square matrix $T_3$ by appending $T_2$ with the $(\al - 1)\times B$ matrix \mbox{$T'_2 = \left[ 0_{(\al-1)\times(k-1)\al} \:\: I_{\al-1} \:\: 0_{(\al-1)\times 1}\right]$}, where $0_{t\times m}$ is a $t\times m$ all-zero matrix and $I_{\al - 1}$ is an $(\al-1)\times(\al-1)$ identity matrix. 

Using the identity matrix $I_{\al-1}$ in $T'_2$, we eliminate all but the last entry in the $k$-th thick-column of $T_2$ to obtain the following matrix:
\begin{equation}
\label{eq:T-semifinal1}
T_4 = 
\begin{bmatrix}
H_{i,1} & \cdots & H_{i,k-1} & 0_{\al\times(\al-1)} & \omega^j\beta_{j\al}^{i-1}\\ 
I_{\al}  & \cdots & 0_{\al} & 0_{\al\times(\al-1)} & \pij{1}{\al}\\ 
\vdots  & \ddots & \vdots & \vdots & \vdots\\ 
0_{\al}  & \cdots & I_{\al} & 0_{\al\times(\al-1)} & \pij{k-1}{\al}\\ 
0_{(\al-1)\times\al}  & \cdots & 0_{(\al-1)\times\al} & I_{\al-1} & 0_{(\al-1)\times 1} 
\end{bmatrix},
\end{equation}
for some $\pij{i}{\al}\in\GF{\qr}^{\al\times 1}$, $1\leq i\leq k-1$.

\begin{figure*}[!h]
\begin{equation}
\label{eq:T-for-l-1}
T' = \twomatrix{H_{\G}}{G_{j,1}^{-1}G_j}(W')^{-1} = 
\begin{bmatrix}
h_{j_1,1} & h_{j_1,2} & \cdots & h_{j_1,\al} & h_{j_1,\al+1} & h_{j_1,\al+2} & \cdots & h_{j_1,B} \\
\vdots & {} & \ddots & {} & \vdots  & {} & \ddots & \vdots  \\
h_{j_g,1} & h_{j_g,2} & \cdots & h_{j_g,\al} & h_{j_g,\al+1} & h_{j_g,\al+2} & \cdots & h_{j_g,B}\\
\omega & 0 & \cdots & 0 & p_{1,1} & p_{1,2} & \cdots & p_{1,B-\al} \\
0 & \omega & \cdots & 0 & p_{2,1} & p_{2,2} & \cdots & p_{2,B-\al} \\
\vdots & {} & \ddots & {}& \vdots  & {} & \ddots & \vdots  \\
0 & 0 & \cdots & \omega & p_{\al,1} & p_{\al,2} & \cdots & p_{\al,B-\al}
\end{bmatrix}
\end{equation}
\hrulefill
\end{figure*}

Now, the determinant of $T_4$ can be written as a polynomial in $\omega$ as follows:
\begin{equation}
\label{eq:det-T-4}
det(T_4) = \omega^j \be_{j\al}^{i-1} + \cdots .
\end{equation}
Note that $\det(T_4)$ is a non-zero polynomial in $\omega$ with coefficients in $\GF{\qr}$, \ie, $\det(T_4) \in \GF{\qr}[\omega]$. Further, the degree of this polynomial is at most $k$. Since $\omega$ is a primitive element of $\GF{\qr^{k+1}}$, the degree of its minimal polynomial is $k+1$. Thus, $\omega$ cannot be a root of any polynomial in $\GF{\qr}$ of degree at most $k$. Hence, we have $det(T_4) \ne 0$. Therefore, $T_4$ is non-singular, and it follows that $T$ must be full-rank.
\end{proof}

{\it Field Size Comparison:} It was shown in~\cite{Silva:09} a parity-check matrix $H$ of a rank-metric code, in particular a Gabidulin code~\cite{Gabidulin:85}, can be used to achieve $g$-weak security for $g\leq 2$. The field size requirement of such a code is $q^B$, where $q$ is the field of the underlying storage code and $B = k\al$ is the total number of information symbols. Since $q\geq n$ for any known MSR code (see \eg,~\cite{YeB:17,Sasidharan:17}, and references therein), the required field size is $O(n^{k\al})$. 

The proposed construction operates over the field size of $\qr^{k+1}$, where $\qr$ is the smallest power of $q$ greater than or equal to $B$. Assuming that $q = O(n)$ and $B=k\al>q$, the proposed construction requires the field size of at most $O\left((nk\al)^{k+1}\right)$. Note that for high-rate MSR codes, the best known codes have $\al$ to be exponential in $k$, and it is shown that $\al$ needs to be at least exponential in $\sqrt{k}$~\cite{Goparaju:14}. Thus, a rank-metric code based outer code would require significantly larger field size as compared to the proposed scheme.

\subsection{Construction for $\ell = 1$ and $g = B - \al$}

Note that $B - \al$ is the maximum value of $g$, as Eve observes $\al$ symbols when $\ell = 1$ (see Remark~\ref{rem:max-g}). The idea of constructing $H$ is similar to Construction~\ref{con:universal-g-1}. In this case, we begin with a Cauchy matrix over some base field and scale its first $\al$ columns by a primitive element of an extension field. The details are as follows.

\begin{construction}
\label{con:universal-l-1}
Consider the parameters of an MSR code as $n, k, d, \al, \be$ and $q$. Choose $\qr$ as the smallest power of $q$ greater than or equal to $2B$, where $B = k\al$. Construct $H$ as the product of two matrices as 
\begin{equation}
\label{eq:universal-H}
H = H' W',
\end{equation}
where $H'$ is a $B\times B$ Cauchy matrix with each entry chosen from $\GF{\qr}$, and $W'$ is a $B\times B$ identity matrix with its first $\alpha$ columns scaled by $1/\omega$. Here, $\omega$ is a  primitive element of the extension field $\GF{\qr^{\alpha+1}}$. We can view $H$ as follows.
\begin{equation}
\label{eq:universal-H-explicit}
H  = 
\begin{bmatrix}
\frac{h_{1,1}}{\omega} & \cdots & \frac{h_{1,\al}}{\omega} & h_{1,\al+1} & \cdots & h_{1,B} \\
\frac{h_{2,1}}{\omega} & \cdots & \frac{h_{2,\al}}{\omega} & h_{2,\al+1} & \cdots & h_{2,B} \\
\vdots & \ddots & \vdots &\vdots & \ddots & \vdots  \\
\frac{h_{B,1}}{\omega} & \cdots & \frac{h_{B,\al}}{\omega} & h_{B,\al+1} & \cdots & h_{B,B} \\
\end{bmatrix},
\end{equation}
where $\left[h_{i,j}\right]$, $1\leq i,j\leq B$, is a Cauchy matrix.
\end{construction}

Next, we show that the above construction can universally achieve $(B-\al)$-weak security for MSR codes when $\ell = 1$.

\begin{theorem}
\label{thm:secure-capacity-universal}
The outer coset code of Construction~\ref{con:universal-l-1} can be used universally with any $(n,k)$ MSR code to store \mbox{$\Bs = k\al$} symbols over $\GF{\qr^{\alpha+1}}$ with $g$-weak security for $g = B - \al$ against an eavesdropper observing any $\ell = 1$ node storing $\al$ coded symbols.
\end{theorem}
%
\begin{proof}
First, it is easy to see that $H$ is non-singular, as it is a $B\times B$ Cauchy matrix with some of its columns scaled by $\omega$. Thus, we have $\Bs = B$. 

Next, we want to show that matrix $T = \twomatrix{H_{\G}}{G_e}$ is full-rank, where $G_e$ is the generator matrix of the observed node and $H_{\G}$ consists of $B - \alpha$ rows of $H$ (see Lemma~\ref{lem:silva}). Consider the case when Eve observes one of the systematic nodes. Since $H'$ is a Cauchy matrix, any of its square sub-matrices is full-rank. Using this property, it is easy to show that $T$ will be full-rank. 

Suppose Eve observes a parity node $j$, $k+1\leq j\leq n$. Note that, for any parity node $k+1\leq j\leq n$, one can easily show that the $\al\times\al$ block $\Gij{j}{1}$ is full-rank as follows. Suppose that a data collector downloads from parity node $j$, $k+1\leq j\leq n$, and systematic nodes 2 through $k$. Since any $k$ out of $n$ nodes allow reconstructing the set of $B$ files, $\Gij{j}{1}$ must be full-rank.

Since $G_{j,1}$ is full-rank, we pre-multiply $G_j$ in $T$ by $G_{j,1}^{-1}$. Then, by multiplying each of the first $\alpha$ columns by $\omega$, we can transform $T$ to the matrix $T'$ shown in~\eqref{eq:T-for-l-1}, for some $\pij{i}{j}\in\GF{\qr}$, $1\leq i \leq \al, 1\leq j\leq B-\al$. 


Now, the determinant of $T'$ can be written as a polynomial in $\omega$ as follows:
\begin{equation}
\label{eq:det-T}
det(T') = \left[det\left(H'_{\G}(\al+1:B)\right)\right]{\omega}^{\al} + \cdots,
\end{equation}
where $H'_{\G}(\al+1:B)$ is the $(B-\al)\times(B-\al)$ sub-matrix of $H'_{\G}$ formed by its last $B-\al$ columns. Since $H'$ is a Cauchy matrix, $det\left(H'_{\G}(\al+1:B)\right) \ne 0$. Hence, $det(T')$ is a non-zero polynomial in $\omega$ with coefficients in $\GF{\qr}$, \ie, $det(T'') \in \GF{q}[\omega]$. Further, $deg(det(T'')) = \alpha$. Since, $\omega$ is a primitive element of $\GF{\qr^{\alpha+1}}$, it cannot be a root of a degree $\alpha$ polynomial in $\GF{\qr}[\omega]$. Therefore, $det(T') \neq 0$, and it follows that $T$ is full-rank.
\end{proof}

{\it Field Size Comparison:} The universal outer code in~\cite{Kurihara:13} based on rank-metric codes achieves $g$-weak security for any $\ell$ and maximum possible $g$. The field size requirement is $q^{2B}$, where $q$ is the field of the underlying storage code and $B = k\al$ is the total number of information symbols. Since $q\geq n$ for any known MSR code (see \eg,~\cite{YeB:17,Sasidharan:17}, and references therein), the required field size is $O(n^{2k\al})$. 

The field size required for the proposed construction is $\qr^{\al+1}$, where $\qr$ is the smallest power of $q$ greater than or equal to $2B$. Assuming that $q = O(n)$ and $2B = 2k\al > q$, the proposed construction requires the field size of at most $O\left((nk\al)^{\al+1}\right)$. Note that for high-rate MSR codes, the best known codes have $\al$ to be exponential in $k$, and it is shown that $\al$ needs to be at least exponential in $\sqrt{k}$~\cite{Goparaju:14}. When $\al$ is exponential in $k$, one can verify that the proposed scheme requires a smaller field size than the rank-metric code based scheme of~\cite{Kurihara:13} for a wide range of parameters.

\section{Conclusion}
\label{sec:conclusion}
We focused on the weak security paradigm in which Eve gains no information about any group of $g$ symbols. We proposed a universal outer code that can weakly secure any MSR code. In particular, we considered two scenarios: (i) the eavesdropper has the maximum strength of $\ell = k - 1$, and the weak security level is the minimum $g = 1$; and (ii) the eavesdropper has the minimum strength of $\ell = 1$, but the weak security level is the maximum possible $g = B - \al$. Our key idea is to utilize the structure present in the (systematic) generator matrix of an MSR code to construct the outer code. This enabled us to reduce the required field size compared to the standard approaches based on rank-metric codes.

\bibliographystyle{IEEEtran}
\bibliography{IEEEabrv,Universal_secrecy}
%
%
%

\end{document}